\documentclass[10pt, conference, letterpaper]{IEEEtran}

\usepackage{endnotes}
\usepackage[T1]{fontenc}
\usepackage{float}
\usepackage{graphicx}
\usepackage{url}
\usepackage{microtype}

\usepackage{amsmath,amsthm,amssymb,latexsym, pifont}

\usepackage[usenames]{color}
\usepackage{ifsym}
\usepackage{wasysym}
\usepackage{booktabs}
\usepackage{mdwlist}
\usepackage{caption}
\usepackage{subcaption}
\usepackage{latexsym}
\usepackage{epstopdf}
\usepackage{color}
\usepackage{enumitem}
\usepackage[ruled]{algorithm2e}

\usepackage{cite}
\usepackage{xspace}

\usepackage{xparse}
\usepackage{nccmath}

\newcommand{\cref}[1]{Chapter~\ref{#1}}

\newcommand{\fref}[1]{Fig.~\ref{#1}}
\newcommand{\tref}[1]{Table~\ref{#1}}

\newcommand{\first}{\emph{(i)}~}
\newcommand{\second}{\emph{(ii)}~}
\newcommand{\third}{\emph{(iii)}~}
\newcommand{\fourth}{\emph{(iv)}~}

\newcommand{\ie}{i.e., \@}
\newcommand{\eg}{e.g., \@}

\newcommand{\cf}{cf. \@}

\newcommand{\etal}{et~al.\xspace}

\newcommand{\perc}{\,\%\xspace}

\definecolor{darkgreen}{rgb}{0,0.5,0}
\definecolor{brown}{rgb}{0.7,0.3,0}
\definecolor{darkblue}{rgb}{0,0,0.5}

\newcounter{fn1}
\newcounter{fn2}
\newcounter{fn3}
\newcounter{fn4}
\newcounter{fn5}
\newtheorem{claim}{Claim}
\newtheorem{definition}{Definition}
\newcommand{\lref}[1]{Claim~\ref{#1}}

\let\underscore\_
\newcommand{\myunderscore}{\renewcommand{\_}{\underscore\hspace{0pt}}}
\myunderscore

\title{Policy-Compliant Path Diversity and Bisection Bandwidth}

\author{\IEEEauthorblockN{Rowan Kl\"oti\IEEEauthorrefmark{1},
Vasileios Kotronis\IEEEauthorrefmark{1},
Bernhard Ager\IEEEauthorrefmark{1},
Xenofontas Dimitropoulos\IEEEauthorrefmark{2}\IEEEauthorrefmark{1}}
\IEEEauthorblockA{\IEEEauthorrefmark{1}ETH Zurich, Switzerland\\
Email: \{rkloeti,~vkotroni,~bernhard.ager\}@tik.ee.ethz.ch\\
\IEEEauthorrefmark{2}University of Crete / FORTH, Greece\\
Email: fontas@ics.forth.gr}}

\begin{document}

\maketitle

\begin{abstract}

How many links can be cut before a network is bisected? 
What is the maximal bandwidth that can be pushed between
two nodes of a network? These questions are closely related
to network resilience, path choice for multipath routing or 
bisection bandwidth estimations in data centers. The answer 
is quantified using metrics such as the number 
of edge-disjoint paths between two network nodes and the
cumulative bandwidth that can flow over these paths. In practice though, such
calculations are far from simple due to the restrictive effect of network policies
on path selection. Policies are set by network administrators to conform to service level agreements,
protect valuable resources or optimize network performance.
In this work, we introduce a general methodology for estimating lower and upper bounds for the
\emph{policy-compliant} path diversity and bisection bandwidth between two nodes of a network, 
effectively quantifying the effect of policies on these metrics. Exact values can be obtained 
if certain conditions hold. The approach is based on regular languages and can be applied in a 
variety of use cases.

\end{abstract}

\section{Introduction}\label{sec:introduction}

Resilience is a desirable property for many networked systems and is often achieved through redundancy:
when multiple paths exist between nodes, it is possible to route around a failed link. This resilience may be 
quantified in a graph-theoretic sense as the (edge-wise) \emph{path diversity}:
\begin{definition}\label{def:path-diversity}
The \emph{path diversity} between two vertices in a graph is the number of edge-disjoint paths connecting them.
\end{definition}
The application of Menger's theorem~\cite{arlinghaus2001graph} 
for edges allows the path diversity between two vertices to be calculated as the minimum cut,
using an algorithm such as Ford-Fulkerson, with each edge having a unitary capacity.
By adding non-unitary edge capacities, we can also calculate the \emph{bisection bandwidth}.
This metric is useful \eg for data center routing to optimize performance and robustness,
or for quantifying the  advantages of multipath routing protocols~\cite{Xu:2006:MMI:1159913.1159934} in terms of
achievable throughput.
\begin{definition}\label{def:bisection-bandwidth}
The \emph{bisection bandwidth} between two vertices in a network is the maximum achievable flow between them.
\end{definition}

In practice, networks do not permit all possible paths due to management policies.
These can be the outcome of routing optimization techniques, security considerations or financial 
agreements~\cite{caesar2005bgp}. For example, inter-domain paths in the Internet resemble a
\emph{valley-free} policy model~\cite{gao2001stable}, a simplified
model of the business relationships between Autonomous Systems (AS).
On the other hand, such policies substantially restrict which paths are permissible and 
constrain the effective path diversity. Thus, a rich graph may not be fully 
utilized due to a restrictive policy (or set of policies) imposed over its paths.
Therefore, calculating the \emph{policy-compliant} path diversity 
and bisection bandwidth is desirable 
to answer questions such as: \first how many \emph{valid} edge-disjoint paths can exist between two nodes, or
\second how much bandwidth can be utilized between these two nodes before the network is overloaded, subject 
to network-wide policies. The goal is to understand the effect of network policy on network resiliency,
availability and achievable throughput.

In this paper, we introduce a method for estimating the path diversity and
bisection bandwidth of a network subject to policy constraints on the paths.
We model the network topology as a directed graph with policy labels on the edges.
We model network policies as a regular expression over these labels and require that all 
valid paths in the graph adhere to this regular expression.
Every regular language can be described by an automaton, specifically a Non-deterministic
Finite state Automaton (NFA)~\cite{bruggemann1993regular}. Using NFAs and the original
graph, we develop a transformed graph that constrains paths to those accepted by the regular
language. While path diversity calculations on the original graph under policies are hard,
they become simple using general graph algorithms~\cite{ford1962flows} on the transformed graph.
For instance, classic graph algorithms can work on the transformed graph to find the 
policy-compliant max-flow/min-cut between two nodes as well as the paths achieving this flow.

We will show how the transformed graph can be utilized to obtain both upper and lower bounds
on the path diversity or bisection capacity of the original graph. If the NFA fulfills
certain criteria, the bounds are equal and therefore exactly the same as the actual value,
\ie the path diversity or bisection bandwidth are invariant under the transformation. Otherwise, 
we obtain upper and lower approximations that encapsulate the actual min-cut within their boundaries.
The tightness of the boundaries depends on the complexity of the state transitions of the NFA,
as we will explain later.
We also show how constraints on traversed nodes may be
imposed, including scenarios where both the nodes and edges are subject to separate
constraints.

The rest of the paper is structured as follows:
Section~\ref{sec:use cases} presents interesting use cases where
path diversity and bisection bandwidth metrics under policy compliance
are required. Section~\ref{sec:idea} describes the basic ideas and the
graph transform process; Section~\ref{sec:graph_transform} substantiates
this process using formal mathematical formulation and proving the needed
claims. Section~\ref{sec:eval} presents some results of our algorithm applied
on Internet AS-level topologies, demonstrating our 
approach\footnote{The source code can be obtained directly from the authors upon request.}.
In Section~\ref{sec:rel-work} we report on related work in the field
of network resilience and policy-compliant min-cuts.
Finally, we conclude the paper and give further outlooks for our work.

\section{Why are Policy-compliant Min-Cuts important?}
\label{sec:use cases}

Calculation of policy-compliant path diversity and bisection bandwidth 
can be applied on a wide range of scenarios to quantify the resilience 
and achievable throughput of a network. 
With our approach, the only requirement is that the network policy should 
be expressible with a regular expression; the form of the corresponding 
NFA dictates whether we can calculate the exact value or an approximation 
as will be described in Section~\ref{sec:graph_transform}. 
Many network policies used in practice are expressible through regular
expressions~\cite{Soule:2013:MNM:2535771.2535792}. We identify the following use cases
for policy-compliant min-cuts:
 
\paragraph*{Inter-Domain Valley-Free Routing} A basic use case is the 
calculation of the path diversity of an Internet conforming to the 
valley-free policy model~\cite{gao2001stable}. Path diversity in this case can refer 
to the number of edge-disjoint paths between two Autonomous Systems (AS), 
where each edge connects two neighboring ASes together. Edges are labelled as \emph{peer-to-peer} ($\mathtt{p2p}$),
\emph{provider-to-customer} ($\mathtt{p2c}$) or \emph{customer-to-provider} ($\mathtt{c2p}$)
relationships. This topology and corresponding edge labels can 
be obtained from datasets like CAIDA~\cite{caida-as-rel}. We note that, in reality, such links
correspond to multiple network layer links and even more physical links, \ie the calculated 
path diversity is therefore a lower bound of the physical path diversity. 
In the valley-free model, the global inter-domain policy can be expressed with the following
regular expression: $\mathtt{c2p^*p2p?p2c^*}$. Paths can only go first uphill
($\mathtt{c2p}$) and then downhill ($\mathtt{p2c}$), while at most one $\mathtt{p2p}$ link
can connect an uphill with a downhill transition forming a ``mountain'' with a $\mathtt{p2p}$
link on its ``peak''. The peak may also be sharper, with a direct transition from uphill to downhill.
We will revisit valley-free path diversity in Section~\ref{sec:eval},
where we examine tier one depeerings.
 
\paragraph*{Beyond Classic Policies - Negative Waypoint Routing}
Valley-free is a basic family of policies that approximates the current
market relationships in the Internet. Ideally, we would like to examine additional 
routing policies on top of the classic ones, across domains. Examples are \first \emph{waypoint routing},
\ie forcing the traffic to pass over certain waypoints before reaching its destination,
and \second \emph{negative routing}, \ie forcing the traffic to avoid certain nodes or links in the network.
Such policies further perplex path diversity calculations but are interesting for specific real-world scenarios
and use cases. 

Consider the following (slightly contrived) example which encapsulates waypoint and negative routing policies. 
We assume a valley-free Internet, in which each
inter-AS edge is directed and is annotated with a tuple label: \emph{(relationship\_type,
next\_AS)}, where \emph{relationship\_type} is $\mathtt{p2p}$, $\mathtt{p2c}$ or
$\mathtt{c2p}$ and \emph{next\_AS} is the edge-terminating AS.
A government organization in AS $\mathtt{A}$ wants to send traffic to one of its embassies 
in AS $\mathtt{B}$, located in another country.
The traffic from $\mathtt{A}$ to $\mathtt{B}$ needs to pass over a special 
encryption middlebox; there are two clones of this middlebox 
in AS $\mathtt{W_1}$ and AS $\mathtt{W_2}$
for redundancy. The original traffic can pass through any AS before it reaches 
waypoint $\mathtt{W_1}$ or $\mathtt{W_2}$, except for 
AS $\mathtt{X}$, which is governed by a rival administration. 
After either of the waypoints is traversed, traffic can go through any AS in the world---including $\mathtt{X}$---until 
it reaches its destination $\mathtt{B}$, where it is decrypted on the premises. 
This policy corresponds to a regular
expression (omitted here for brevity) which can be mapped in turn
to an NFA that accepts the expression, as depicted in Fig.~\ref{fig:neg_route_nfa}. 
The resulting policy-compliant NFA is then a composition of the valley-free
and the negative waypoint routing NFAs. 

In this particular case, the path diversity metric can help the traffic sender determine how
many inter-domain links could be brought down until the organization-to-embassy
communication is crippled, \eg in the event of a cyber-war launched by the rival country.
We view this as a ``negative waypoint'' inter-domain routing use case,
since we aim at approximating the multitude of edge-disjoint paths that \emph{avoid} 
a certain node or edge in the inter-AS graph and pass through certain waypoints. The 
problem cannot be solved by simply removing AS $\mathtt{X}$, pruning its corresponding 
edges and calculating the diversity on the pruned graph using the sender-waypoint
and waypoint-receiver node pairs. As the policy dictates, AS $\mathtt{X}$ \emph{can be
traversed} after the traffic has been processed by \emph{either} AS $\mathtt{W_1}$ or AS $\mathtt{W_2}$, but not sooner.
The NFA encompasses both this stateful routing decision process and the underlying
valley-free conditions, and allows us to encompass a complex policy in a simple
graph object. Having such NFAs at hand, we will explain later how the policy-compliant
diversity can be calculated. 

\begin{figure}
\centering
\includegraphics[width=\columnwidth,clip=true,trim=0 0 0 0]{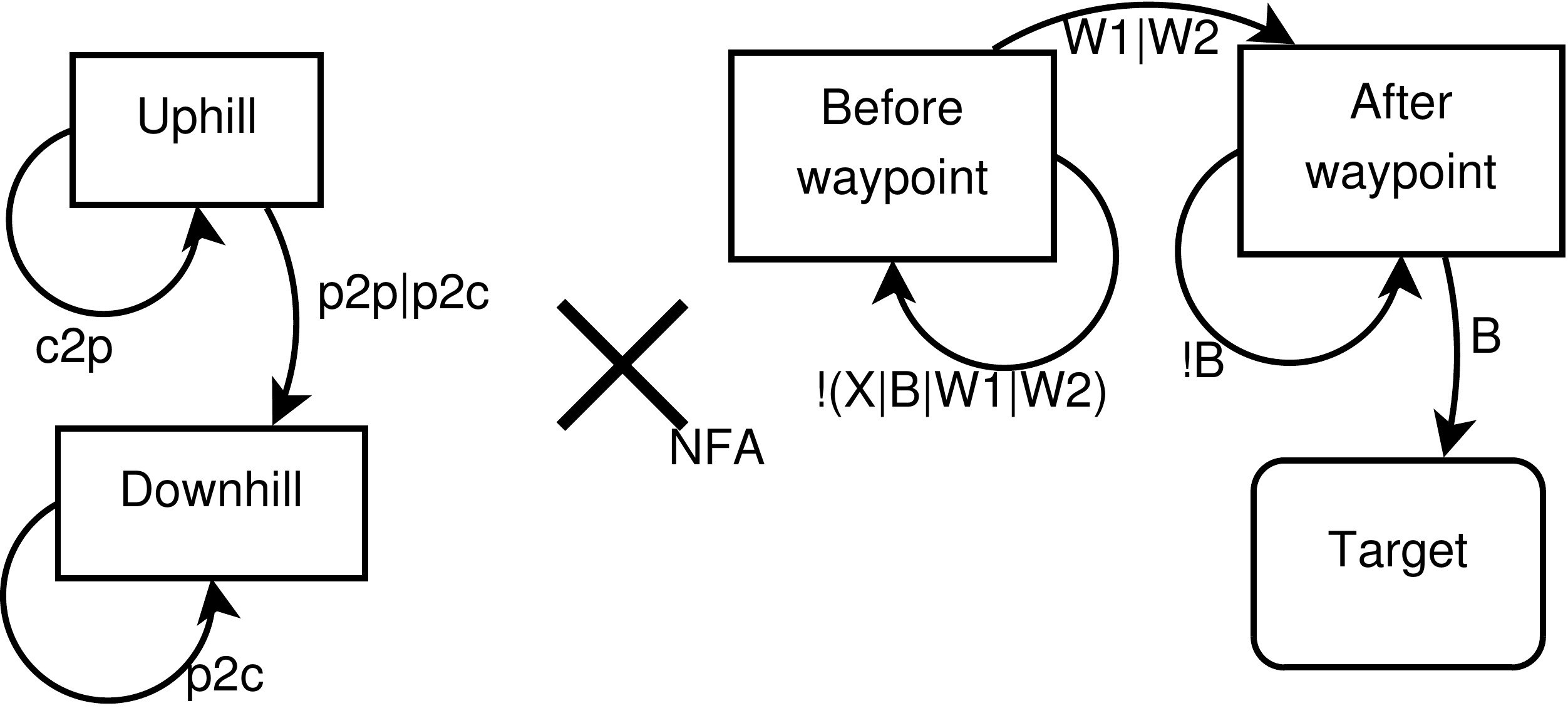}
\caption{NFAs for the negative waypoint routing use case: valley-free (left side) and waypoint routing policies (right side).}
\label{fig:neg_route_nfa}
\end{figure}

\paragraph*{Multipath TCP (MPTCP)} MPTCP~\cite{4032726} is a proposed extension to TCP from the IETF~\cite{rfc6824} 
allowing TCP connections to use multiple paths to increase resource utilization, redundancy
and availability. This is especially useful when multiple wireless channels with different properties are
available, or for better utilizing dense data center topologies, to exploit the large
bisection bandwidth available. Consider the data center scenario, where the operator 
has full control over the endpoints and switches and can route subflows individually.
Here, the number of disjoint available paths that MPTCP can send traffic over is useful 
information for the MPTCP implementation. Path diversity calculations in this case yield an approximation of the
maximal number of distinct flows that MPTCP can push in the network without these flows
contending with each other for bandwidth.
The approach for calculating the effective path diversity can take into account domain-specific
policies which have been, \eg expressed via Merlin~\cite{Soule:2013:MNM:2535771.2535792}. In addition, 
it enables data center operators to compute the 
policy-compliant available bandwidth between two areas of a network. This can help
to estimate the time that a MPTCP bulk data transfer may require, or
whether the network is utilized properly during the transfer.

\paragraph*{DDoS Link-Flooding Attacks}  Estimating the bisection
bandwidth of a data center can lead to an approximation of the attack budget that 
a DDoS link-flooding attack against data center core links, such as 
Crossfire~\cite{Kang:2013:CA:2497621.2498106}, may require. The recent attack 
against Spamhaus~\cite{prince2013ddos} indicates
that such information is valuable both for the attacker and the defender: the attacker tries 
to find weak links which he can deplete at minimal cost,
isolating entire domains from the Internet~\cite{Kang:2013:CA:2497621.2498106},
while the defender tries to increase the cost for the attacker via suitable network and
traffic engineering~\cite{chow2007distributed}. Knowing how much bandwidth needs to be
depleted to cut off a network from the rest of the Internet can help an operator perform
an informed risk assessment of a possible attack, while taking into account the routing policies
imposed over the network.

\section{Idea: Custom Graph Transform}\label{sec:idea}

Our objective is to transform a directed labelled graph $G$ into another graph $G'$,
such that: \first only policy-compliant paths exist in $G'$, and \second the transformation
does not distort the minimum cut. The minimum cut may represent either bisection bandwidth or,
by choosing unitary edge capacities, path diversity (\cf Menger's theorem~\cite{arlinghaus2001graph}).
We define a \emph{policy-compliant path} as any path whose \emph{path string} (the string resulting
from concatenating the edge labels) belongs to some regular language $L$. How do we accomplish this?
Every regular language $L$ is represented by a finite state automaton $M$, and vice versa.
Therefore, upon traversing an edge in the graph, we need to accordingly change the state in the automaton,
as if the edge label had been given to the automaton as an input.

The first idea is to use the \emph{tensor product} of the graph $G$ and the policy-checking
automaton $M$ (\fref{fig:tensor-product} gives an example).
From $M$ we define the transition graph $T$, which has a node set consisting
of the states of $M$ and an edge set representing the permitted state transitions of $M$.
The edges in $G$ and $T$ are labelled with policy-related symbols that result in the state
transitions represented in $T$. 
For each of these symbols $s \in \Sigma$, we form the subgraphs $G_s$ and $T_s$ consisting
of all the edges labelled with $s$. We then form the \emph{tensor product} of these
subgraphs, which is defined as follows: if $v_1$ and $v_2$ are nodes in $G$ and $q_1$ and
$q_2$ are nodes in $T$, then $G'$ contains the nodes $(v_1,q_1)$, $(v_1,q_2)$, $(v_2,q_1)$
and $(v_2,q_2)$. Likewise, if there is an edge $(v_1,v_2)$ in the subgraph $G_s$ and an
edge $(q_1,q_2)$ in the subgraph $T_s$, then the tensor product contains the edge
$((v_1,q_1),(v_2,q_2))$. The union of all of these tensor products is the transformed graph
$G'$. This allows us to move both between the nodes $v_1$ and $v_2$ in $G$ and the nodes
$q_1$ and $q_2$ in $T$ (and therefore the states $q_1$ and $q_2$ in $M$) at the
``same time''. We define $M$ to be a \emph{Non-deterministic} Finite state
Automaton (NFA) because an NFA is typically smaller than a corresponding \emph{Deterministic}
Finite state Automaton (DFA). 

This transform gives us part of the solution, but it does not guarantee that the
minimum cut calculated over $G'$ represents accurately the corresponding minimum cut of $G$.
A single edge in the original graph $G$ may be mapped to a set of parallel edges in the transformed graph $G'$.
Consider the possibility of a transition from one state $\lbrace q_1 \rbrace$ to two states
$\lbrace q_2,q_3 \rbrace$ over an edge $(v_1,v_2)$. This will be mapped to two
edges in $G'$: $((v_1,q_1),(v_2,q_2))$ and $((v_1,q_1),(v_2,q_3))$, even though there is just one edge in $G$.
Therefore, the second idea is to \emph{force} the transform to maintain the minimum cut. 
This may not always be possible, as we will explain later. 
If it is not, the result still is an approximation of the minimum cut, for which
we will give lower and upper bounds in section~\ref{sec:correctness}.

\begin{figure}[t!]
\centering
\begin{subfigure}[b]{0.49\columnwidth}
\includegraphics[width=\columnwidth,clip=true,trim=0 0 0 0]{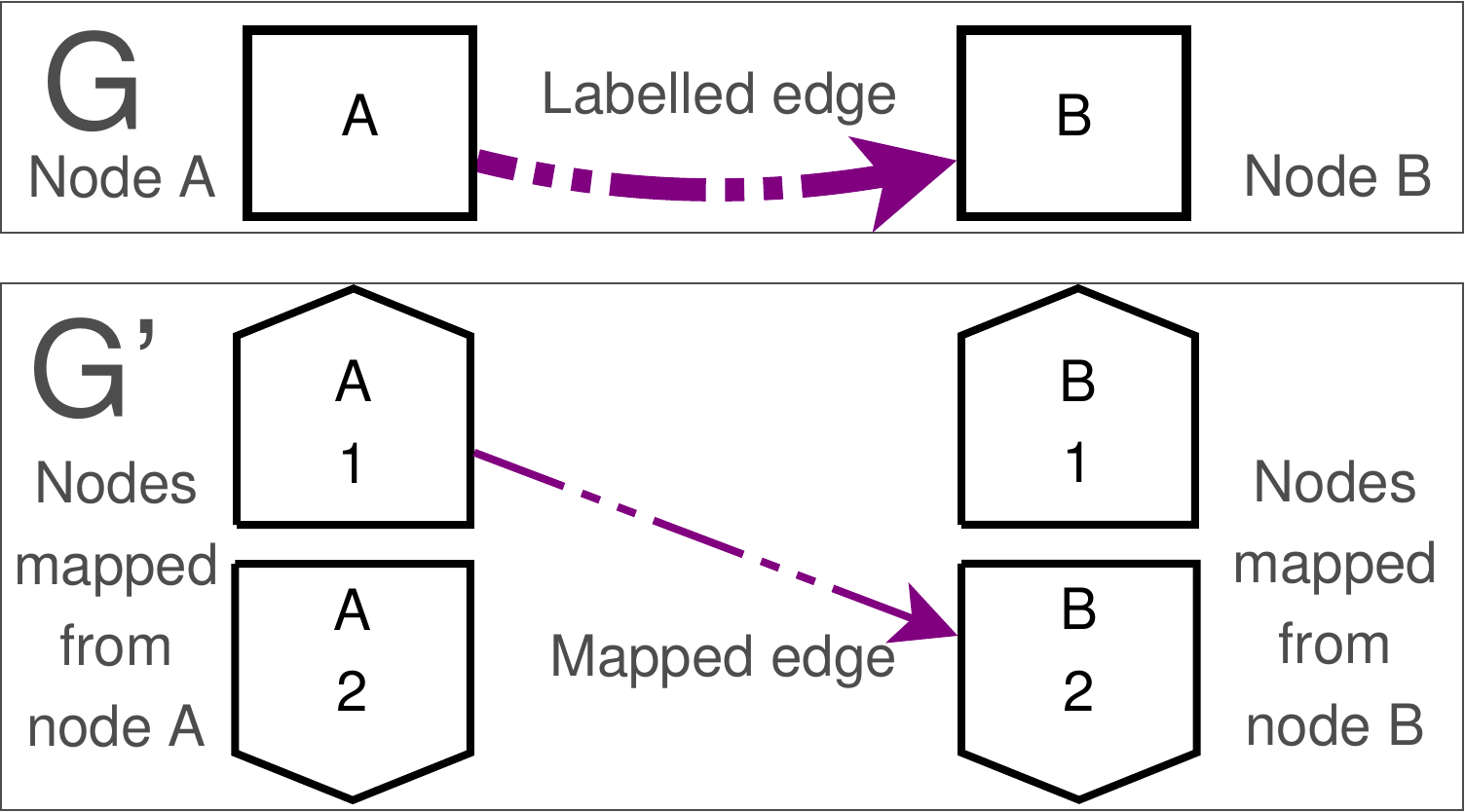} 
\caption{Original$\rightarrow$Transformed Graph}\label{fig:tensor-product-graph}
\end{subfigure}
\begin{subfigure}[b]{0.49\columnwidth}
\includegraphics[width=\columnwidth,clip=true,trim=0 0 0 0]{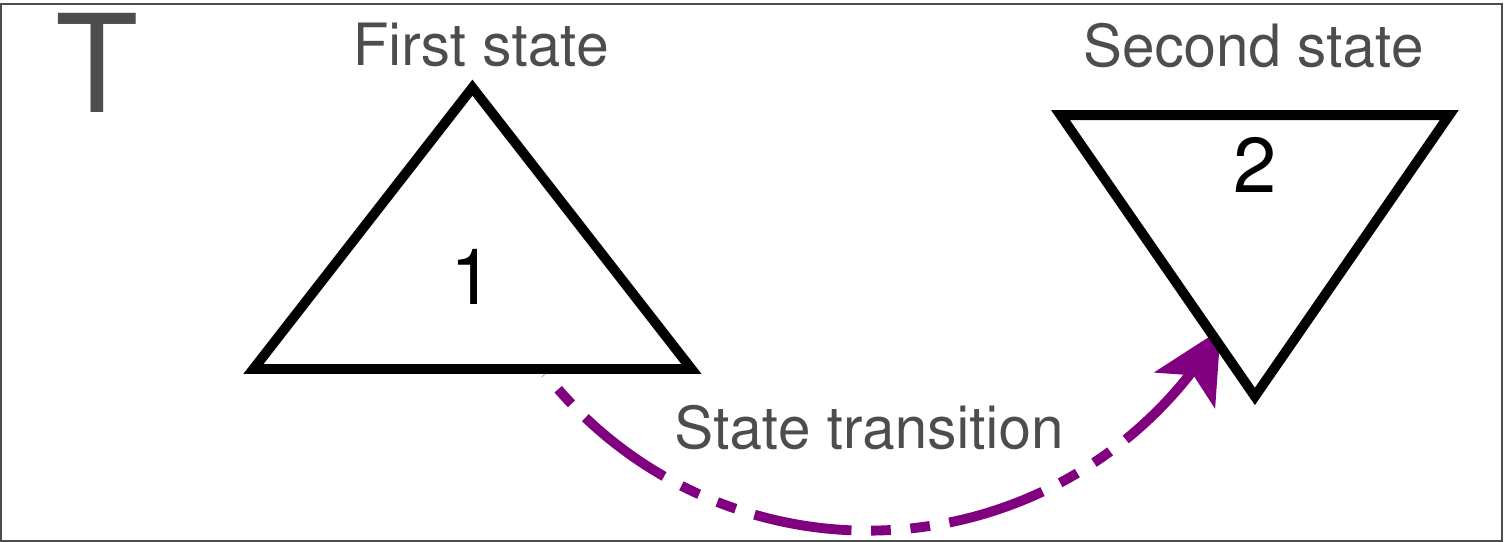} 
\caption{NFA}\label{fig:tensor-product-nfa}
\end{subfigure}
\caption{The tensor product. For each label $s \in \Sigma$, the tensor product of the edges
in $G$ labelled with $s$ and the edges in $T$ labelled with $s$ (representing the state
transitions in $M$ when $s$ is entered as an input) yields the set of mapped edges in $G'$.
We represent the edge label $s$ symbolically via the colour and style of the edge. Nodes in $G'$
are derived from the triangles of the NFA and the squares of the original graph $G$.}
\label{fig:tensor-product}
\end{figure}

The core idea is to add \emph{aggregator} states to the NFA and consequently to $G'$ so that
the min-cut paths between two nodes must traverse at most the same number of parallel edges
as in $G$, which limits the minimum cut in $G'$ to the same value as in $G$. 
To preserve the structure of the NFA---and thus the policy it describes---we utilize 
\emph{$\varepsilon$-transitions}. \emph{$\varepsilon$-transitions}  
can be thought of as ``free'' transitions: they do not consume a symbol during traversal. 
In our case, this means that we
do not need to traverse an edge between nodes in $G$ in order to traverse an
$\varepsilon$-transition. 
Where there is a chance of the min-cut being inflated, we add an
aggregator node (for each edge label) and use $\varepsilon$-transitions to ``channel'' all
of the paths through this node. Each node in $G$ has corresponding aggregator
nodes in $G'$ \emph{as needed} and where \emph{it is applicable}. 

There are several possible cases for the aggregated transitions: \first from a single state
to another single state (\emph{one-to-one}), \second from a single state to multiple states
(\emph{one-to-many}), \third from multiple states to a single state (\emph{many-to-one}),
or \fourth from multiple states to multiple states (\emph{many-to-many}). In the first case,
the min-cut is always invariant; no inflation can occur.
The latter cases are depicted in \fref{fig:divergent}, \fref{fig:convergent},
and \fref{fig:many-many}, respectively. In the \emph{one-to-many} and \emph{many-to-one} cases 
the addition of aggregation nodes leads to the correct min-cut value, while the \emph{many-to-many} 
case requires careful consideration. 
In this case, we can only maintain the accurate min-cut using aggregation states if
the set of state transitions can be expressed as the Cartesian product of two subsets of the set of state nodes. 
This results in complete bipartite subgraphs on the transformed graph $G'$
with aggregatable transitions, as shown in \fref{fig:many-many}.
Therefore, if this is not the case, we need to break the transition down to $n$ disjoint state transition 
sets of the cases \first to \fourth and perform the transform for each of them; each set introduces an 
extra aggregator state.
Regarding edge capacities we provide different values to yield upper and lower bounds on 
the min-cut, as we will show later.

\begin{figure}[t!]
\centering
\begin{subfigure}[b]{0.49\columnwidth}
\includegraphics[width=\columnwidth,clip=true,trim=50 250 45 50]{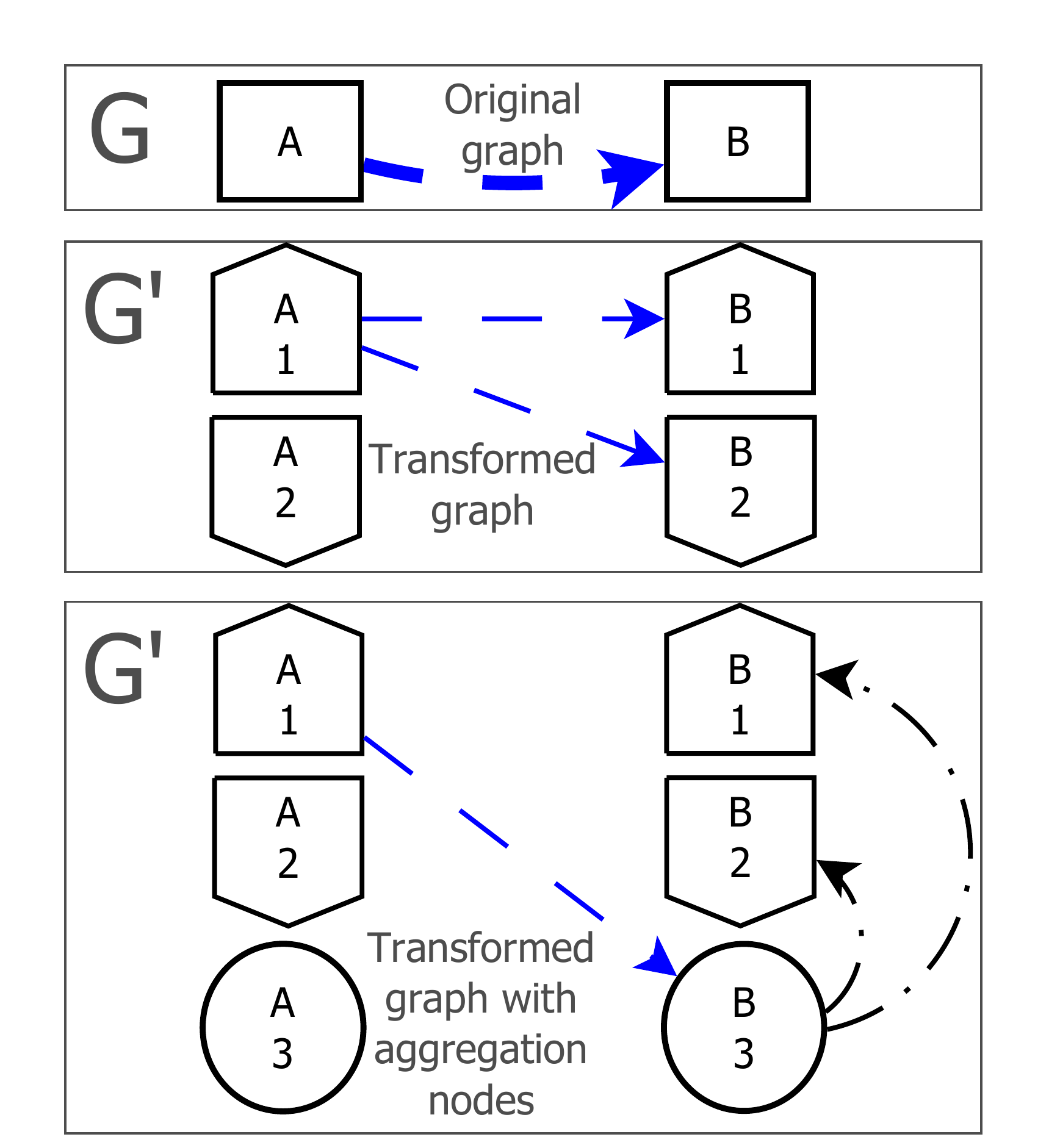} 
\caption{Original$\rightarrow$Transformed Graph}\label{fig:aggregator-nodes-divergent}
\end{subfigure}
\begin{subfigure}[b]{0.49\columnwidth}
\includegraphics[width=\columnwidth,clip=true,trim=0 0 0 0]{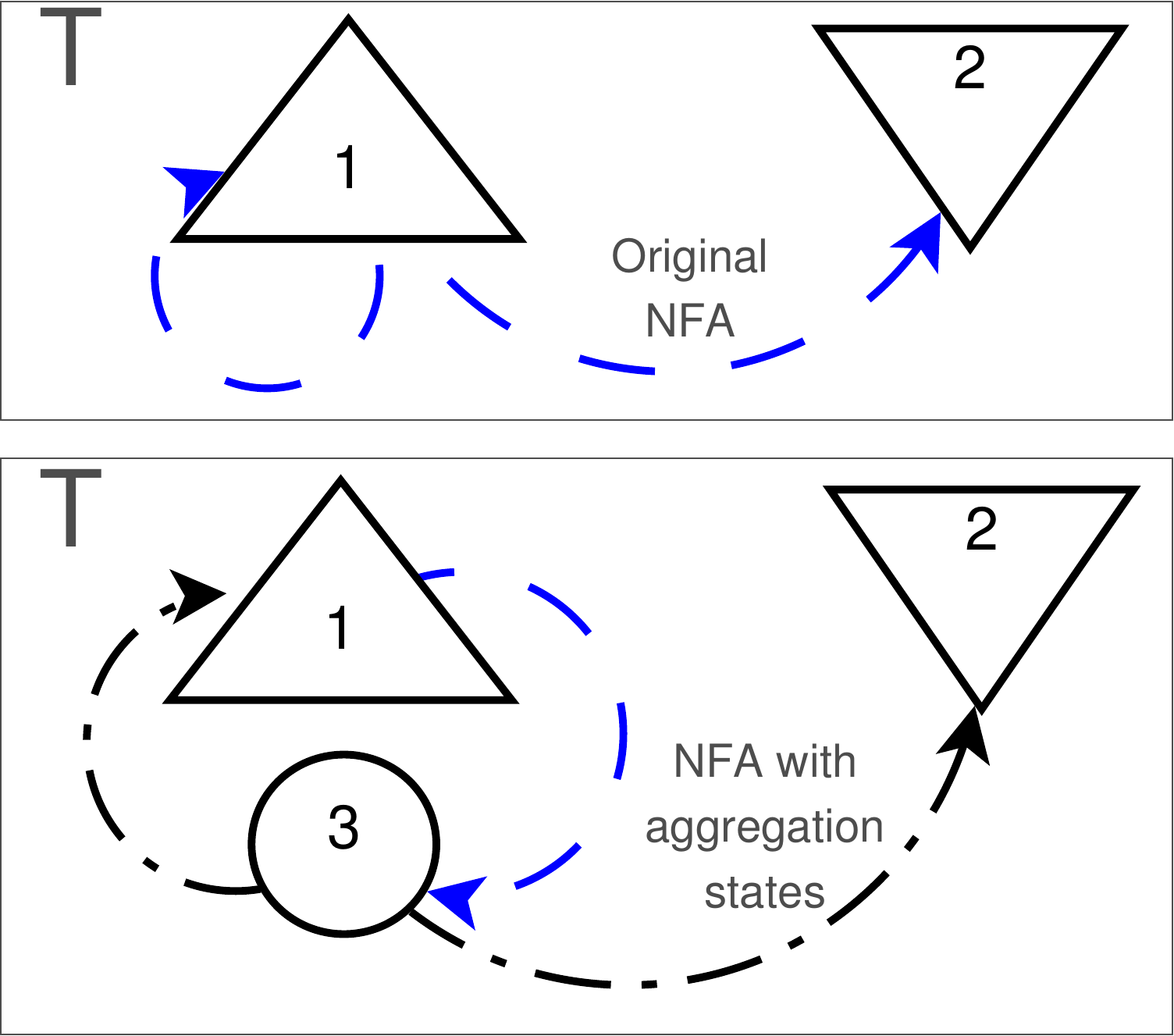} 
\caption{NFA}\label{fig:nfa-aggregator-nodes-divergent}
\end{subfigure}
\caption{The use of aggregation nodes (represented with circles) for a \emph{one-to-many} state mapping.
The two blue dashed transitions of the NFA are aggregated using one aggregation node. 
$\varepsilon$-transitions are represented by black dotted dashed lines.}
\label{fig:divergent}
\end{figure}

\begin{figure}[h!]
\centering
\begin{subfigure}[b]{0.49\columnwidth}
\includegraphics[width=\columnwidth,clip=true,trim=50 250 45 50]{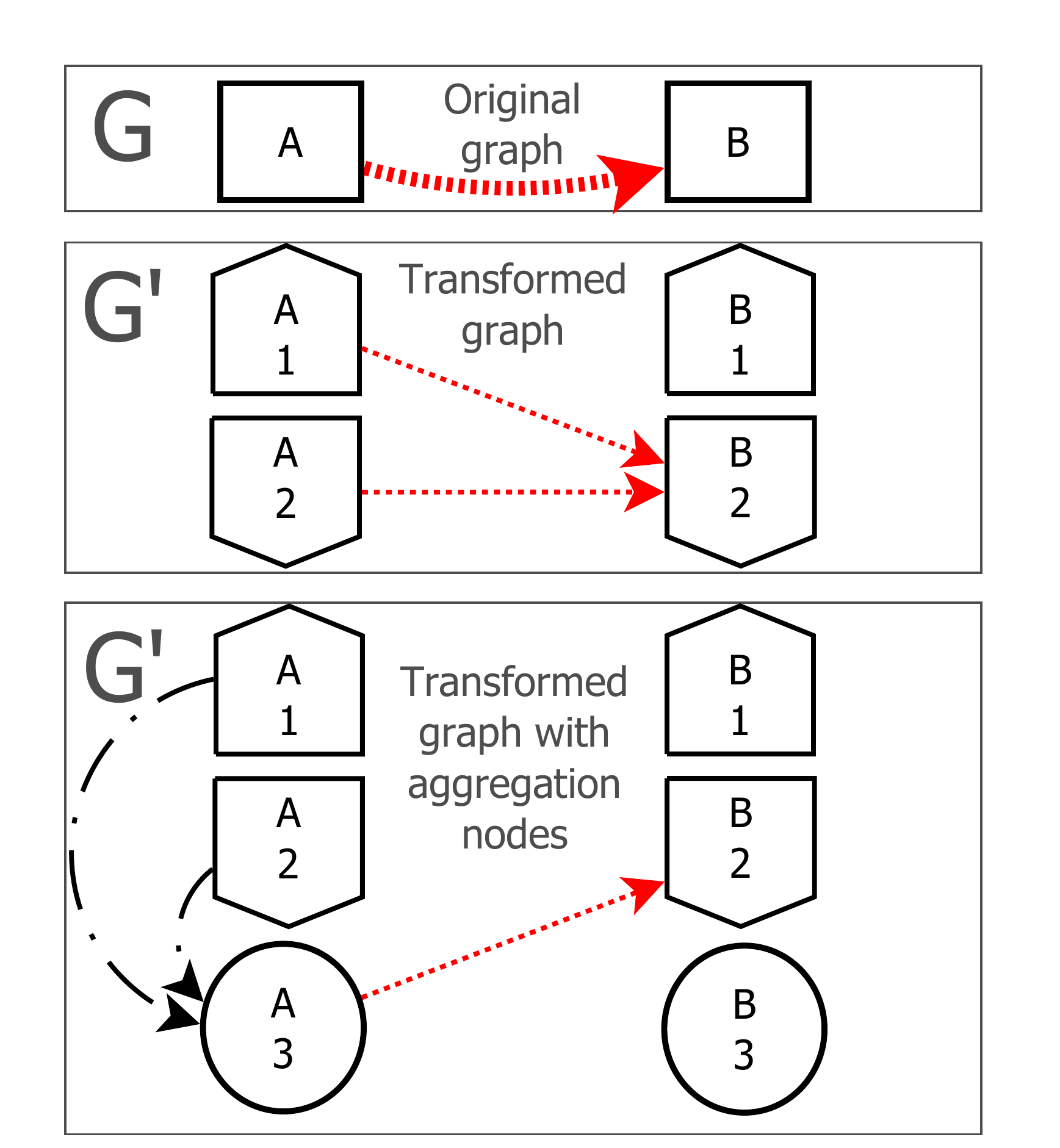} 
\caption{Original$\rightarrow$Transformed Graph}\label{fig:aggregator-nodes-convergent}
\end{subfigure}
\begin{subfigure}[b]{0.49\columnwidth}
\includegraphics[width=\columnwidth,clip=true,trim=0 0 0 0]{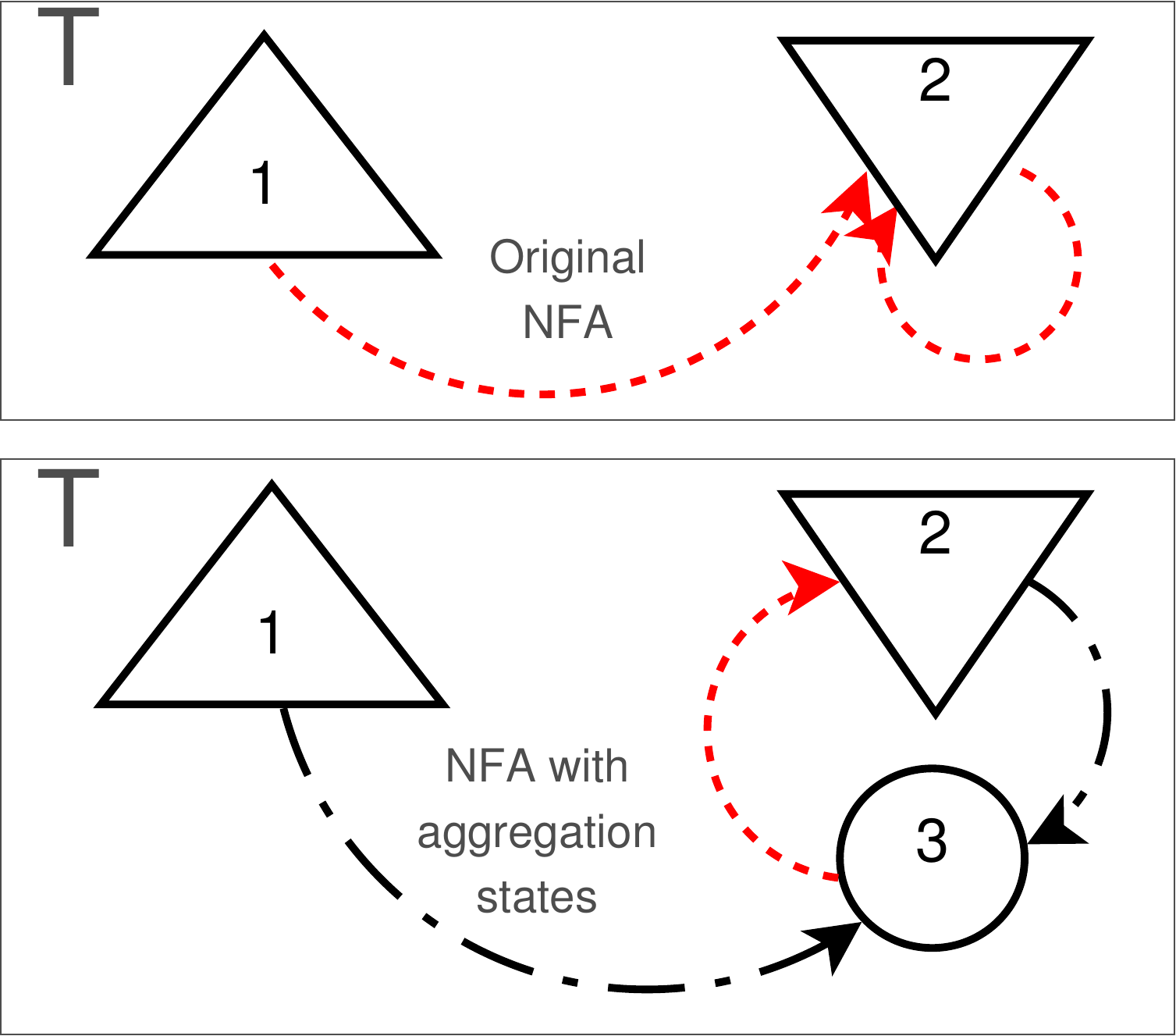} 
\caption{NFA}\label{fig:nfa-aggregator-nodes-convergent}
\end{subfigure}
\caption{The use of aggregation nodes (represented with circles) for a \emph{many-to-one} state mapping.
The two red dotted transitions of the NFA are aggregated using one aggregation node.}
\label{fig:convergent}
\end{figure}

\begin{figure}[h!]
\centering
\begin{subfigure}[b]{0.49\columnwidth}
\includegraphics[width=\columnwidth,clip=true,trim=50 150 45 50]{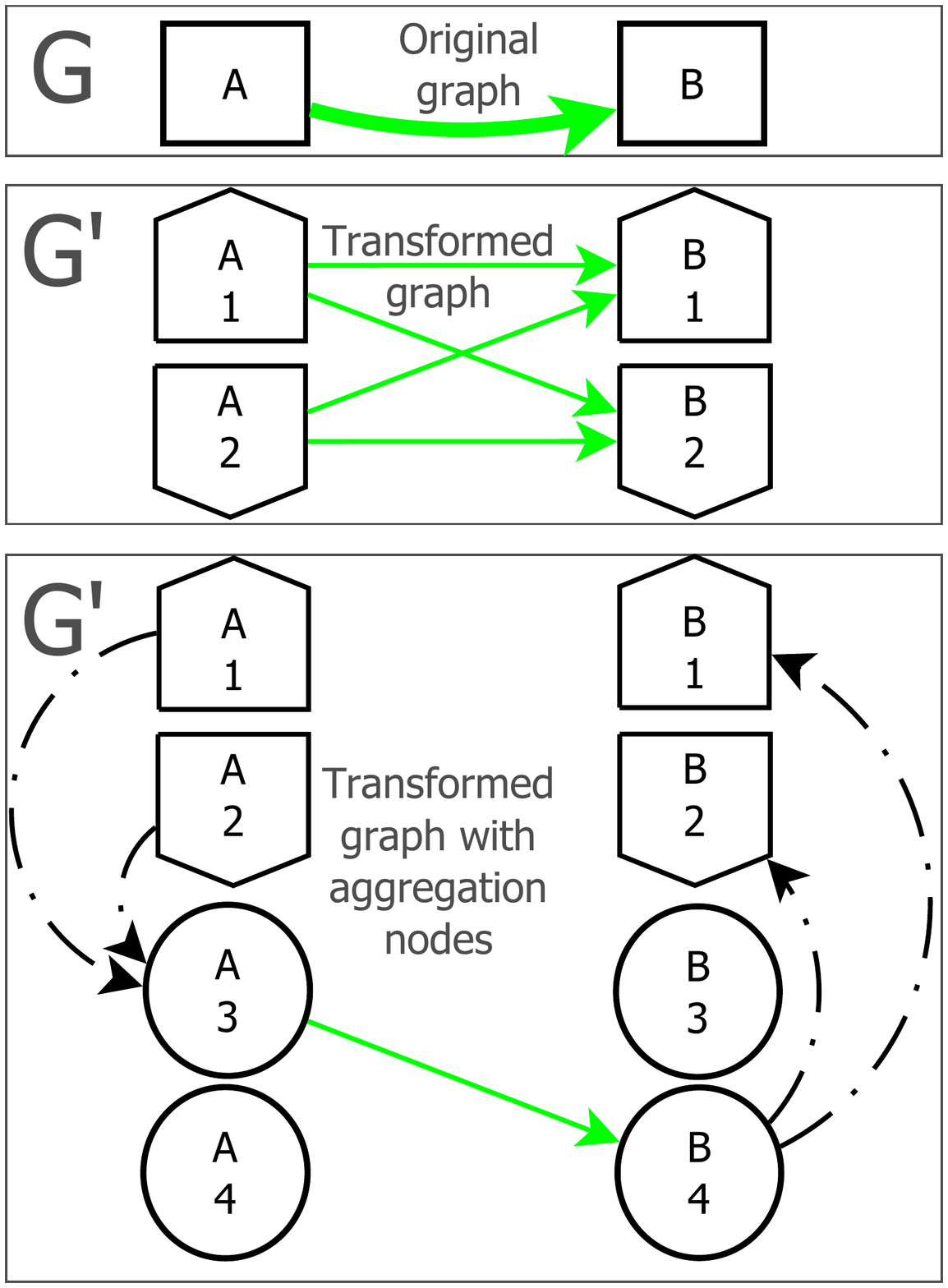} 
\caption{Original$\rightarrow$Transformed Graph}\label{fig:many-many-transformed}
\end{subfigure}
\begin{subfigure}[b]{0.49\columnwidth}
\includegraphics[width=\columnwidth,clip=true,trim=0 0 0 0]{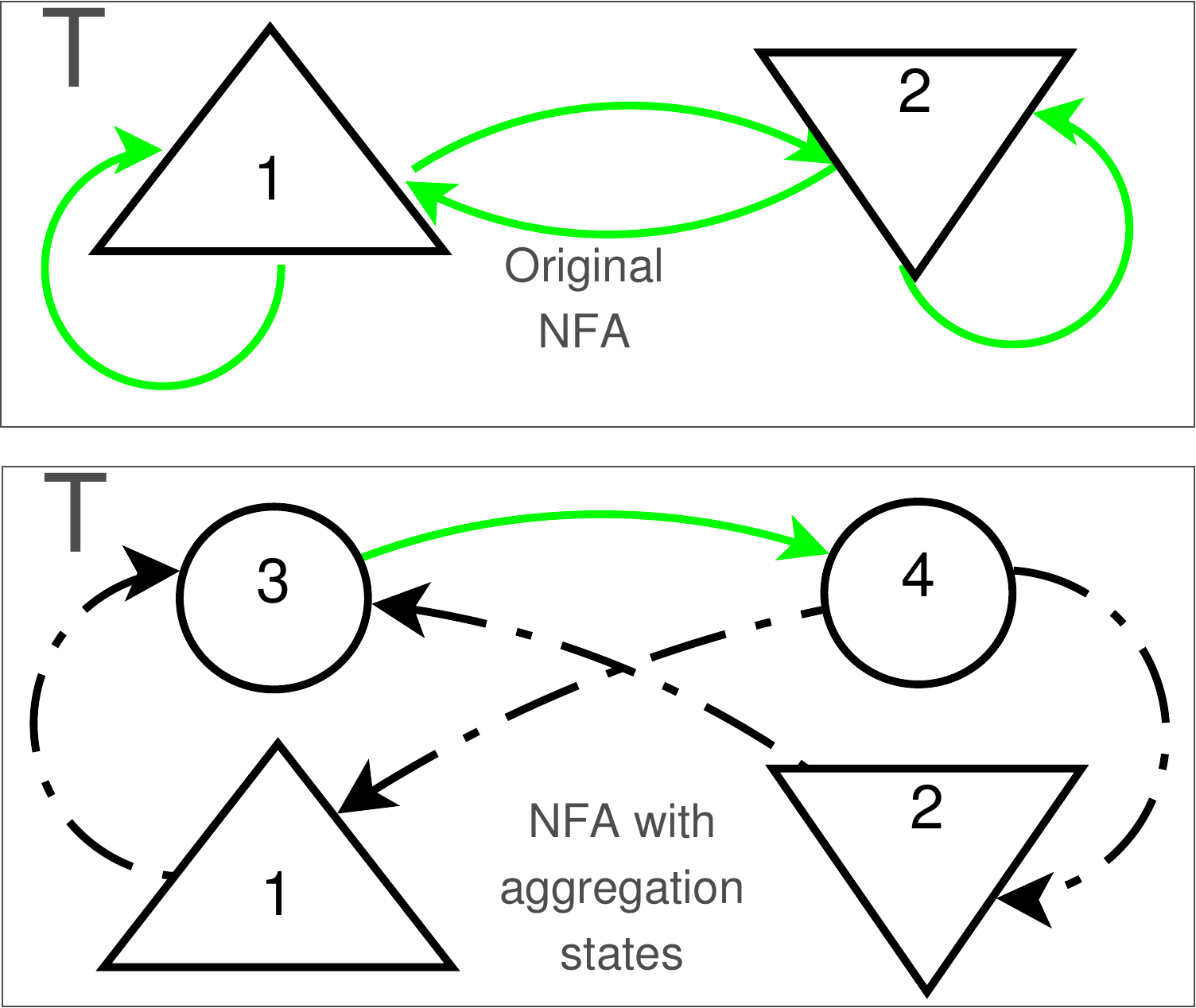} 
\caption{NFA}\label{fig:many-many-NFA}
\end{subfigure}
\caption{The use of aggregation nodes for a \emph{many-to-many} state mapping.
The four green transitions of the NFA are aggregated using two aggregation nodes.}
\label{fig:many-many}
\end{figure}

To calculate the min-cut with conventional algorithms, we need to have a single
terminating node in $G'$, which implies a single terminating state in $M$. If we have more
than one terminating state, then prior to proceeding we need to create a virtual terminating
state and copy all of the transitions to the previous terminating states to the new one. 
Finally, if we want to consider constraints on nodes as well as edges, we can split the
nodes under consideration in two halves as follows: all of the incoming edges connect to one
half and all of the outgoing edges connect to the other. We then add a labelled edge between
those halves, allowing the edge to represent the node and effectively encompass its constraints.

\section{The Math Behind the Transform}\label{sec:graph_transform}

Here we present the required definitions, algorithmic steps, 
proofs of correctness, and complexity of the algorithm.

\subsection{Notation and Definitions}\label{sec:definitions}
\begin{itemize}

\item $G = (V,E)$ is a labelled directed graph with nodes $V$, edges $E$, a labelling
function $l: E \rightarrow \Sigma$ that maps edges in $E$ to corresponding labels in
$\Sigma$, and a capacity function $c: E \rightarrow \mathbb{R}^+$ that maps edges to edge
capacities. For path diversity calculations, we choose $c(e) = 1$ for all $e \in E$.

\item $L \subseteq \Sigma^{*}$ is a regular language. We require that for any path
$(e_1,...,e_n)$ in $G$, the string $l(e_1)...l(e_n)$ formed of the edge labels be in
$L$.

\item $M = (Q,\Sigma,\Delta,q_{0},F)$ is a NFA with states $Q$, input symbols $\Sigma$,
state transitions $\Delta \subseteq Q \times \Sigma \times Q$, terminating states
$F \subseteq Q$ and a  starting state $q_{0} \in Q$. $M$ accepts $L$.
$\varepsilon$-transitions are allowed with $\varepsilon \in \Sigma$.
\item $T = (Q,\Delta)$ is the directed graph derived from $M$.
\item $G' = (V',E')$ is the final transformed graph: its formation is such that 
the policy constraints are met and the minimum cut is not inflated, if possible.

\end{itemize}

\subsection{Single Virtual Termination States}\label{sec:termination}
In order to have a single terminating node in $G'$ for use with Ford-Fulkerson (or another
flow algorithm), the set of terminating states $F$ in $M$ must be no larger than one.
If $|F| > 1$, we add a new termination state $q^*$. A state transition to $q^*$ is added
from every other state which previously had a transition to a terminating state, and $F$ is updated accordingly:
\begin{equation}\label{eq:prestep1}
Q := Q \cup \lbrace q^* \rbrace
\end{equation}
\begin{equation}\label{eq:prestep2}
\Delta := \Delta \cup \lbrace (q',s,q^*) \mid (q',s,q) \in \Delta \wedge q \in F \rbrace
\end{equation}
\begin{equation}\label{eq:prestep3}
F := \lbrace q^* \rbrace
\end{equation}

\subsection{Tensor Product Transform: Steps}\label{sec:transform}

\begin{enumerate}

\item We form subgraphs $G_{s} \subseteq G$ consisting of all the edges labelled with
$s \in \Sigma$:
\begin{equation}\label{eq:step1}
G_{s} := (V,E_{s})
\end{equation}
\begin{equation}\label{eq:step2}
E_{s} := \lbrace (v_{1},v_{2}) \mid (v_{1},v_{2}) \in E \wedge l((v_{1},v_{2})) = s \rbrace
\end{equation}

\item We form subgraphs $T_{s} \subseteq T$ of all the state transitions for an input symbol
$s \in \Sigma$:
\begin{equation}\label{eq:step3}
T_{s} := (Q,\Delta_{s})
\end{equation}
\begin{equation}\label{eq:step4}
\Delta_{s} := \lbrace (q_{1},q_{2}) \mid (q_{1},s',q_{2}) \in \Delta \wedge s' = s \rbrace
\end{equation}

\item We augment $\Delta_{s}$ with aggregator states, if necessary. \first We require that
$\Delta_{s}$ has the form $Q'_s \times Q''_s$ with $Q'_s \subseteq Q$ and
$Q''_s \subseteq Q$. \second If this is not the case, 
we decompose $\Delta_s$ into $n_s$ disjoint sets $\Delta_{s,k}$ such that each subset has
the form $Q'_{s,k} \times Q''_{s,k}$ and repeat the following for every $k$:

\begin{enumerate}

\item We add a new aggregator state $q'_s$ if $\vert Q'_s \vert > 1$ and $q''_s$ if
$\vert Q''_s \vert > 1$:
\begin{equation}\label{eq:step5a1}
Q := \left\{ \begin{array}{l l}
Q & \vert Q'_s \vert = 1 \\
Q \cup \lbrace q'_s \rbrace & \vert Q'_s \vert > 1
\end{array} \right.
\end{equation}
\begin{equation}\label{eq:step5a2}
Q := \left\{ \begin{array}{l l}
Q & \vert Q''_s \vert = 1 \\
Q \cup \lbrace q''_s \rbrace & \vert Q''_s \vert > 1
\end{array} \right.
\end{equation}

\item If we added a $q'_s$ in the previous step, we connect the preceding states to it with
an $\varepsilon$-edge. Likewise for $q''_s$ and succeeding states:
\begin{equation}\label{eq:step5b1}
\Delta_{\varepsilon} := \left\{ \begin{array}{l l}
\Delta_{\varepsilon} & \vert Q'_s \vert = 1 \\
\Delta_{\varepsilon} \cup (Q'_s \times \lbrace q'_s \rbrace) & \vert Q'_s \vert > 1
\end{array} \right.
\end{equation}
\begin{equation}\label{eq:step5b2}
\Delta_{\varepsilon} := \left\{ \begin{array}{l l}
\Delta_{\varepsilon} & \vert Q''_s \vert = 1 \\
\Delta_{\varepsilon} \cup (\lbrace q''_s \rbrace \times Q''_s) & \vert Q''_s \vert > 1
\end{array} \right.
\end{equation}

\item Finally, we connect states $q'_s$ and $q''_s$ with the aggregated state transition edge. If we did not use aggregating
nodes on either side because there was only one state in $Q'_s$ or $Q''_s$, we use that single
state instead:
\begin{equation}\label{eq:step5c1}
Q'_s := \left\{ \begin{array}{l l}
Q'_s & \vert Q'_s \vert = 1 \\
\lbrace q'_s \rbrace & \vert Q'_s \vert > 1
\end{array} \right.
\end{equation}
\begin{equation}\label{eq:step5c2}
Q''_s := \left\{ \begin{array}{l l}
Q''_s & \vert Q''_s \vert = 1 \\
\lbrace q''_s \rbrace & \vert Q''_s \vert > 1
\end{array} \right.
\end{equation}
\begin{equation}\label{eq:step5c3}
\Delta_s = Q'_s \times Q''_s
\end{equation}

\end{enumerate}

\item For each pair of $G_{s}$ and $T_{s}$ that we derived in the previous steps, we calculate the
\emph{tensor product} $G'_{s} = G_{s} \times T_{s}$. The nodes  of $G_s'$ are given by:
\begin{equation}\label{eq:step6a}
V' := V \times Q
\end{equation}
The edges of $G_s'$ are given by:
\begin{align}\label{eq:step6b}
E'_s := \lbrace ((v_{1},q_{1}),(v_{2},q_{2})) \nonumber \\ \mid (v_{1},v_{2}) \in E_{s} \wedge (q_{1},q_{2}) \in \Delta_{s} \rbrace
\end{align}
Furthermore, $\varepsilon$-transitions are mapped to edges in $G'$ that are effectively located within a single node in
$G$:
\begin{equation}\label{eq:step6c}
E'_{\varepsilon} := \lbrace ((v,q_{1}),(v,q_{2})) \mid v \in V \wedge (q_{1},q_{2}) \in \Delta_{\varepsilon} \rbrace
\end{equation}

\item $G'$ is the union of $G'_{s}$ for each $s \in \Sigma$ (including $\varepsilon$):
\begin{equation}\label{eq:step7}
G' := \bigcup_{s \in \Sigma} G'_{s} = (V',\bigcup_{s \in \Sigma} E'_s)
\end{equation}
\item For the edge capacities $c': E' \rightarrow \mathbb{R}^+$ in $G'$, we have different values for
the upper and lower bound minimum cut. \eqref{eq:step8a} and \eqref{eq:step8b} give the
capacities $c'_{upper}$ and $c'_{lower}$ for the upper and lower bound minimum cuts,
respectively. 
\begin{equation}\label{eq:step8a}
c'_{upper}(((v_1,q_1),(v_2,q_2))) = \left\{ \begin{array}{l l}
c(v_1,v_2) & v_1 \neq v_2 \\
\infty & v_1 = v_2 \\
\end{array} \right.
\end{equation}
\begin{equation}\label{eq:step8b}
c'_{lower}(((v_1,q_1),(v_2,q_2))) = \left\{ \begin{array}{l l}
\frac{c(v_1,v_2)}{n_{s}} & v_1 \neq v_2 \\
\infty & v_1 = v_2 \\
\end{array} \right.
\end{equation}
Where $s = l((v_1,v_2))$ and $n_s$ is the number of disjoint
sets $\Delta_{s,k}$ that $\Delta_s$ is decomposed into in step 3.

\end{enumerate}

\subsection{Correctness}\label{sec:correctness}

We next prove that the transformed graph contains only valid
policy-compliant paths (claims~\ref{thm:reverse-correctness} and \ref{thm:forward-correctness}),
that capacities $c'_{lower}$ and $c'_{upper}$ yield lower and upper 
bounds of the min-cut (claims~\ref{thm:min-cut-lower-bound} and \ref{thm:min-cut-upper-bound}), 
and that we obtain an
exact value for the policy-compliant min-cut as long as certain 
conditions hold, pertaining to the form of the NFA (claim~\ref{thm:min-cut-exact-value}).

\begin{claim}\label{thm:reverse-correctness}
Given a path $P$ in $G$, if the string formed by the concatenation of the edge labels of $P$
is not in $L$, then no corresponding path $P'$ exists in $G'$.
\end{claim}
\begin{proof}
The edge labels of $P$ form a string. This string contains at least one edge $e$ with a
label $s \in \Sigma$ which results in the string no longer being in $L$. This implies that
there is no outgoing state transition from the preceding state to any other state in the
NFA. The tensor product of the NFA with the edge $e$ is thus empty. Therefore, there is
no edge to the next node mapped from $P$ and no $P'$ can be formed in $G'$.
\end{proof}

\begin{claim}\label{thm:forward-correctness}
If a path $P$ exists in $G$ and the string formed by the concatenation of the edge labels of
$P$ is in $L$, then there exists a corresponding path $P'$ in $G'$.
\end{claim}
\begin{proof}
If the edge labels traversed by $P$ form a string in $L$, then that string represents a
sequence of valid state transitions in NFA $M$ (respectively the NFA graph $T$) 
from the starting to a terminating state. The tensor product for each $s \in \Sigma$ gives
us a connection between two nodes if the edge corresponds to a valid state transition in
$M$. As the string is in $L$, we know that all of the edge transitions are valid and
therefore that all of the nodes mapped from $P$ to $P'$ are connected. Thus a valid path $P'$ can
be formed in $G'$.
\end{proof}

\begin{claim}\label{thm:min-cut-lower-bound}
Let $v_1$ and $v_n$ be nodes in $G$. Let $q_0$ and $q_t$ be the starting and terminating
states in the NFA. Let the edge capacities of $G'$ be $c'_{lower}$ of \eqref{eq:step8a}.
Then the minimum cut between $(v_1,q_0)$ and $(v_n,q_t)$ in $G'$ is less than or equal to
the minimum cut between $v_1$ and $v_n$ in $G$, taking into consideration only those paths
whose edge labels form strings in $L$.
\end{claim}
\begin{proof}
From \lref{thm:reverse-correctness} and \lref{thm:forward-correctness} we have that any
path in $G'$ corresponds to a valid path in $G$, and vice versa. For each pair of adjacent
nodes $v_k$ and $v_{k+1}$ in the path in $G$ we have the capacity $c((v_k,v_{k+1}))$ for the
edge that connects them. The edge $(v_k,v_{k+1})$ is mapped to $n_s$ edges in $G'$, each
having a capacity of $\frac{c((v_k,v_{k+1}))}{n_s}$, where $n_s$ is the number of disjoint
sets $\Delta_{s,k}$ that $\Delta_s$ is decomposed into in step 3. All of these $n_s$ mapped 
edges in $G'$ therefore have a cumulative capacity of $c((v_k,v_{k+1}))$, the same as
between the pair of nodes in $G$ that they were mapped from. Hence the minimum cut between
each pair of nodes in the path in $G'$ is at most as large as that in $G$, while it may also
be smaller due to the $\Delta_s$ decomposition. By induction, this applies to the path as a
whole, and by generalization to all paths in $G'$. Therefore, the minimum cut will not be
overestimated and the calculated value in $G'$ is a lower bound of the actual min-cut in
$G$.
\end{proof}

\begin{claim}\label{thm:min-cut-upper-bound}
 Let $v_1$ and $v_n$ be nodes in $G$. Let $q_0$ and $q_t$ be the starting and terminating
states in the NFA. Let the edge capacities of $G'$ be $c'_{upper}$ of \eqref{eq:step8b}.
Then the minimum cut between $(v_1,q_0)$ and $(v_n,q_t)$ in $G'$ is greater than or equal to
the minimum cut between $v_1$ and $v_n$ in $G$, taking into consideration only those paths
whose edge labels form strings in $L$.
\end{claim}
\begin{proof}
From \lref{thm:reverse-correctness} and \lref{thm:forward-correctness} we have that any
path in $G'$ corresponds to a valid path in $G$, and vice versa. For each pair of adjacent nodes
$v_k$ and $v_{k+1}$ in the path in $G$ we have the capacity $c((v_k,v_{k+1}))$ for the edge
that connects them. The edge $(v_k,v_{k+1})$ is mapped to $n_s$ edges in $G'$ which all have the capacity 
$c((v_k,v_{k+1}))$, where $n_s$ is the number of disjoint
sets $\Delta_{s,k}$ that $\Delta_s$ is decomposed into in step 3.  Hence the
capacity of an edge $e$ in $G$ and of any edge $e'$ in $G'$ that $e$ is mapped to is the same. All
valid paths in $G'$ therefore have at least the same minimum cut as the ones in $G$ from which they are
mapped, while there may be several corresponding parallel paths in $G'$ due to the $\Delta_s$ decomposition. Therefore, the 
minimum cut will not be underestimated and the calculated value in $G'$ is an upper bound of the actual min-cut in $G$.
\end{proof}

\begin{claim}\label{thm:min-cut-exact-value}
The lower bound of claim~\ref{thm:min-cut-lower-bound} and the upper bound of
claim~\ref{thm:min-cut-upper-bound} coincide and the min-cut calculation is exact 
if:
\begin{equation}\label{eq:correctness}
\forall s \in \Sigma: \exists Q'_s,Q''_s \subseteq Q: \Delta_s = Q'_s \times Q''_s
\end{equation}
\end{claim}
\begin{proof}
If \eqref{eq:correctness} is true, then $n_s$ is equal to one for every $s \in \Sigma$.
Accordingly, $c'_{lower}(e) = c'_{upper}(e)$ for every $e \in E'$. This means that the lower
and upper bounds of the minimum cut are equal. Since the actual min-cut lies between these
values, it must therefore be equal to the lower and upper bounds.
\end{proof}

\subsection{The Maximal Biclique Generation Problem}

The number of disjoint sets $\Delta_{s,k}$ that $\Delta_s$ is decomposed into in step 3 of the graph transform
determines how far off the lower and upper min-cut bounds are from the actual value in the worst case.
Each $\Delta_{s,k}=Q'_{s,k} \times Q''_{s,k}$ should be expressed as the Cartesian product of two
state node subsets. The tensor product of each $\Delta_{s,k}$ with a link $(v_i, v_j) \in E_s$ is equivalent to a 
biclique in the transformed graph. 
This process can thus be reduced to finding the minimal number of complete bipartite 
graphs---or bicliques---that cover the transformed subgraph corresponding to the initial labelled link $(v_i, v_j)$.
The found bicliques can then be used inversely to determine the $k$ $\Delta_{s,k}$ sets, \ie to determine
the Cartesian product decompositions of the NFA state transitions. This observation helps to estimate the
complexity of the decomposition problem; though a formal analysis is outside the scope of this paper.
The problem of finding maximal bicliques is generally $\mathsf{NP}$-Complete~\cite{Peeters:2003:MEB:958763.958769}, and we refer 
the reader to existing literature for solutions~\cite{DBLP:conf/colognetwente/Kayaaslan10,BinkeleRaible201064}. If looser bounds
are acceptable, a heuristic solution could be used.
We typically only need to decompose small subgraphs corresponding to simple NFAs and this only needs to be performed
once per $\Delta_s$---\ie per symbol--- when transforming the graph. The time costs are significant only 
in cases where the NFA contains very large and complicated transitions. For many common scenarios, 
including 1-to-1, 1-to-N, N-to-1, and M-to-N, the decomposition is trivial (see Fig.~\ref{fig:tensor-product} to~\ref{fig:many-many}).

\subsection{Algorithmic Complexity}

\paragraph*{In Space} Here we consider the space complexity, which depends
on the number of nodes and edges in graph $G$ and the number of states and transitions in NFA $M$.
There are $\vert Q \vert$ states in the NFA. We may need to add $\mathcal{O}(\vert \Delta \vert)$ states
for aggregation (at most one per transition). Therefore, applying the tensor product and taking into account the
aggregation states gives us a total node complexity of:
\begin{equation}\label{eq:node-complexity}
\vert V' \vert = \mathcal{O}(\vert V \vert (\vert Q \vert + \vert \Delta \vert))
\end{equation}
Typically, one edge in $G$ will be mapped to one edge in $G'$, plus some
$\varepsilon$-edges. In the worst case, we may need $\mathcal{O}(\vert \Delta \vert)$
edges to map between two nodes (if many need to be decomposed into disjoint subsets in step 3). 
We may also need $\mathcal{O}(\vert \Delta \vert)$ $\varepsilon$-edges for each node, 
yielding a total edge complexity of:
\begin{equation}\label{eq:edge-complexity}
\vert E' \vert = \mathcal{O}(\vert \Delta \vert(\vert V \vert + \vert E \vert))
\end{equation}

\paragraph*{In Time} 
To obtain $G_s$, executing \eqref{eq:step2} requires
$\Theta(\vert E \vert)$ steps (the nodes are maintained), while obtaining
the $T_s$ requires executing \eqref{eq:step4}, requiring
$\Theta(\vert \Delta \vert)$ steps. Step 3 will be executed in the worst
case $n$ times, where $n$ is the number of disjoint transition
sets that $\Delta_s$ may be broken into, which cannot be larger than $\vert \Delta \vert$.
Additionally, it demands $t_{dec}$ time, which is the amount of time required to 
actually decompose state transitions in $M$. The latter depends on $\Delta$, and may
be of nondeterministic polynomial complexity. However, $t_{dec}$ is generally negligible in practice 
for many common scenarios (\ie $\vert V \vert >> \vert Q \vert$ and $\vert E \vert >> \vert \Delta \vert$).
3a and 3c require only constant time, since they add a single object to a set.
3b requires $\mathcal{O}(\vert Q \vert)$ steps, giving us
$\mathcal{O}(\vert Q \vert \vert \Delta \vert)$ total time for step 3. For
step 4, we have $\mathcal{O}(\vert V \vert \vert Q \vert)$ for executing \eqref{eq:step6a}, 
$\mathcal{O}(\vert E \vert \vert \Delta \vert)$ for executing  \eqref{eq:step6b} and
$\mathcal{O}(\vert V \vert \vert \Delta \vert)$ for executing \eqref{eq:step6c}.
Finally, \eqref{eq:step7} requires $\mathcal{O}((\vert V \vert + \vert E \vert) \vert \Delta
\vert)$ steps. Thus, the total time complexity is:
\begin{equation}\label{eq:time-complexity}
t = \mathcal{O}(\vert V \vert \vert Q \vert + \vert \Delta \vert(
\vert V \vert + \vert E \vert + \vert Q \vert)) + t_{dec}
\end{equation}
In practice, the total running time is dominated by the min-cut calculation on the transformed graph---\eg via Ford-Fulkerson 
or Edmonds-Karp---rather than by the graph transformation process itself. The spatial complexity of the transformed graph in terms of the sizes
$\vert V' \vert$ and $\vert E' \vert$ is the most important factor for the time that the min-cut calculation requires. 

\section{Applications}
\label{sec:eval}
In this section we showcase some of the applications of the graph transform algorithm.
Note that our main contribution is the algorithm; here, we simply
wish to demonstrate its applicability to a selection of real-world problems, rather 
than provide a full and rigorous analysis.

\subsection{Paths between tier one and tier two providers}

\paragraph*{Setup} 
We begin by calculating the number of edge-disjoint paths between some of the largest ISPs
in the world, using the CAIDA AS relationships data~\cite{caida-as-rel}. 
We performed this calculation both for the total
path diversity assuming no policies, \ie \emph{arbitrary} paths, and also for the diversity
of valley-free paths using the regular expression described in
Section~\ref{sec:use cases}. We note that this
regular expression can be mapped to an NFA that satisfies condition~(\ref{eq:correctness});
therefore, the obtained values in this case are exact since the min-cut bounds coincide. 

\paragraph*{Results} 
The results are shown in \tref{tab:as-path-diversity}. The ISPs are: NTT (AS2914),
Deutsche Telekom (AS3320), AT\&T (AS7018), Embratel (AS4230), BT (AS5400) and Comcast (AS7922).
First of all, we notice that unconstrained routing could offer around
two orders of magnitude greater path diversity than the valley-free case. We observe 
substantial path diversity present between tier ones and tier twos, especially considering that
these are AS-level paths, corresponding to multiple links at the router level.
We note that the number of valley-free paths between pairs of tier one ISPs is always one
due to the lack of an upstream ISP as well as the prohibition of traversing multiple peering
links; thus we see only their direct \emph{p2p} interconnections. The sizeable customer cone
of the tier ones combined with this lack of diversity means that any depeering has the
potential to cause major disruption for the direct tier one customers, which is known to
have occurred already~\cite{underwoodwrestling, brown2008peering}. For large tier twos, due to the rich 
peer-to-peer interconnectivity resulting in large path diversity, such depeerings seem not
to be harmful. 
In this context we remark that even for very large ISPs, the limiting factor for
path diversity is often the number of peering and provider connections that the ISP itself
maintains, rather than the Internet topology at large; this points to a densely connected 
Internet. An ISP which wishes to improve its connectivity can therefore either establish 
a business relationship with another upper tier ISP, or expand its peering. Many evidently choose 
the latter~\cite{Chatzis:2013:MIM:2541468.2541473},
with the propagation of public open Internet Exchange Points (IXPs)~\cite{euroix-ixp-list} lowering the
barriers to entry for establishing new peering relationships.

\begin{table}
\centering
\scriptsize
\tabcolsep2pt
\begin{tabular}{lcccccccccccc}
\toprule
 & \multicolumn{6}{c}{Arbitrary paths} & \multicolumn{6}{c}{Valley-free paths}\\
 \cmidrule{2-7} \cmidrule{8-13}
AS & \textbf{2914} & \textbf{3320} & \textbf{7018} & 4230 & 5400 & 7922 & \textbf{2914} & \textbf{3320} & \textbf{7018} & 4230 & 5400 & 7922\\\midrule
\textbf{2914} & -    & 496 & 1012 & 190 & 145 & 145 & - & 1  & 1 & 9  & 5  & 2\\
\textbf{3320} & 496  & -   & 496  & 190 & 145 & 145 & 1 & -  & 1 & 10 & 6  & 3\\
\textbf{7018} & 1012 & 496 & -    & 190 & 145 & 145 & 1 & 1  & - & 9  & 5  & 3\\
        4230  & 190  & 190 & 190  & -   & 145 & 145 & 9 & 10 & 9 & -  & 10 & 8\\
        5400  & 145  & 145 & 145  & 145 & -   & 145 & 5 & 6  & 5 & 10 & -  & 6\\
        7922  & 145  & 145 & 145  & 145 & 145 & -   & 2 & 3  & 3 & 8  & 6  & -\\
\bottomrule
\end{tabular}
\caption{Path diversity between large tier one and tier two ISPs. Tier one ISPs
are marked in bold face for clarity.}
\label{tab:as-path-diversity}
\end{table}

\subsection{Global path diversity}
\paragraph*{Setup} To extend the previous use case, we look at the number of edge-disjoint
paths between arbitrary pairs of ASes. We again use the CAIDA AS relationships
data~\cite{caida-as-rel} and use the valley-free model. Due to the large number of
ASes present, it is not feasible to calculate the pairwise path diversity for every pair
of ASes. Instead we sample the ASes, selecting each AS with a probability proportional to
the number of addresses it announces, using the CAIDA RouteViews AS-to-prefix
data~\cite{caida-routeviews}. The AS relationships dataset is obtained from BGP routes
announced at various vantage points within the Internet; the downside to this is that many
peering links are not visible from these vantage points. In reality, the number of peering links in the
Internet may be much larger than reported in this
dataset~\cite{Ager:2012:ALE:2377677.2377714}. To address this deficiency, we augment the AS relationships 
with data available from PeeringDB, which was evaluated and validated
by Lodhi \etal~\cite{peeringdb-routing-ecosystem}.
We specifically add links between IXP members with an \emph{Open} policy, indicating a general
willingness to peer with other IXP members without any preconditions (such as balanced
traffic). The majority of IXP members have an Open peering policy. In addition to the valley-free
scenario, we also evaluate a more liberal policy, where a path may traverse more than one peering
link between the uphill and downhill transitions. As we will see later, this can increase resiliency substantially.
We note that the multiple peering link case corresponds also---like the classic valley-free case---to an NFA 
that satisfies condition~(\ref{eq:correctness}),
thus yielding exact values for the path diversity calculations.

\paragraph*{Results}

\fref{fig:path-diversity-global} shows the CCDF of path diversity. We evaluated the path
diversity between 10,000 pairs of randomly selected ASes, with the selection weighted by the
number of announced addresses as already described beforehand. 
We observe that \first the added links make little difference
to the valley-free scenario, \second the multiple peering links scenario has a considerably
larger path diversity and profits a little more from extra peering links.
Note that although the two valley-free scenarios appear to be identical, there are small differences
between them hidden by the log scale, which we will see later. 

\begin{figure}
\centering
\includegraphics[width=0.75\columnwidth,clip=true,trim=0 0 0 0]{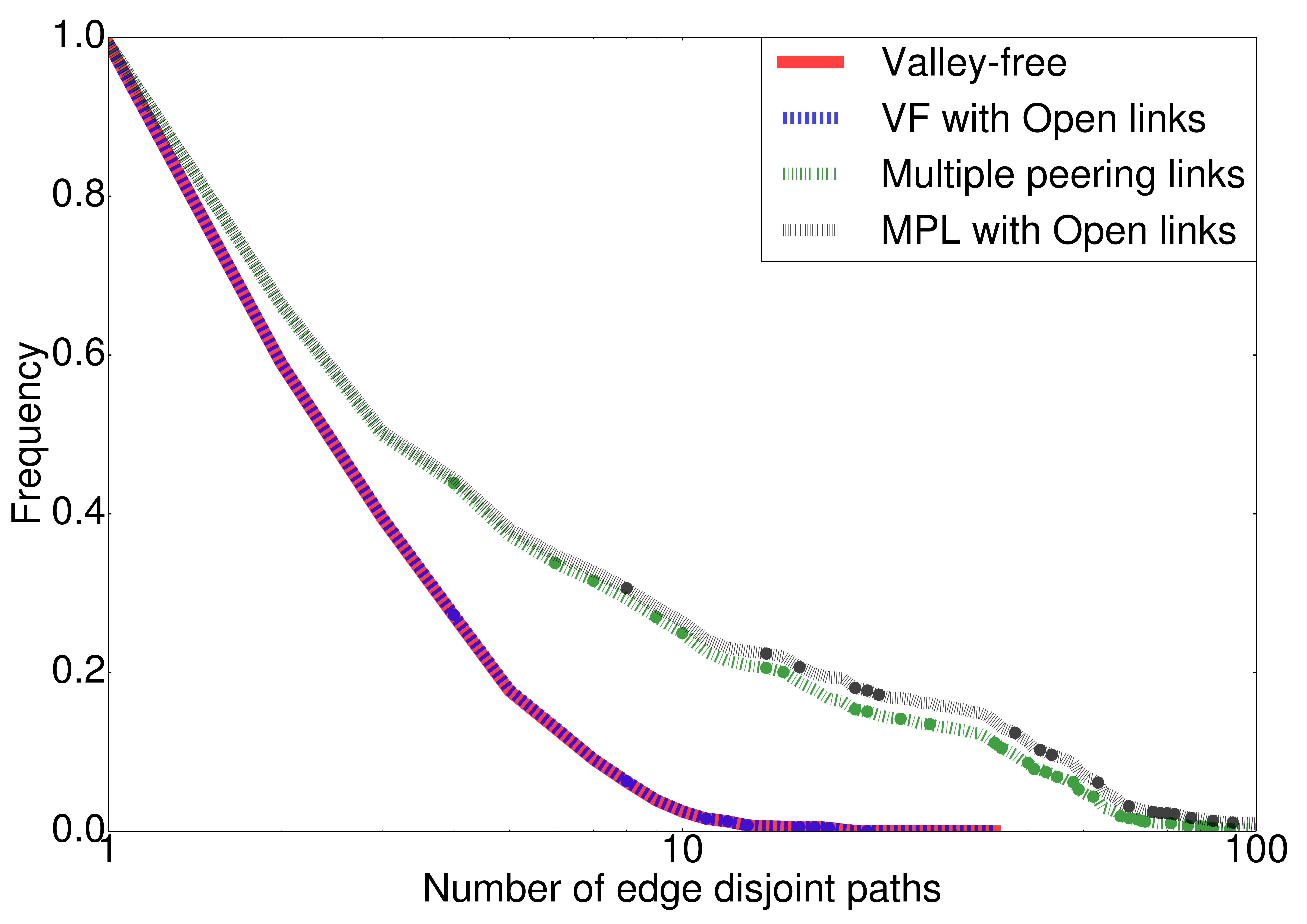} 
\caption{CCDF of path diversity, by constraint type and added link policy. VF: Valley-Free, MPL: Multiple Peering Links.}
\label{fig:path-diversity-global}
\end{figure}

\subsection{Effects of peering policy openness on path diversity}
\paragraph*{Setup}
Typically, IXP members have \emph{Open}, \emph{Selective} and \emph{Restrictive} peering
policies. The majority of IXP members have an Open policy, which, as stated above, implies
the willingness to peer with other IXP members without any preconditions. A Selective policy
generally implies certain preconditions, such as balanced traffic (\eg a maximal 1:2 ratio)
or geographically diverse peering points (\eg one on every continent), and is typical for
major ISPs. A Restrictive peering policy generally indicates the intent to peer only on a
case-by-case basis, which is typical for the largest ISPs, especially tier ones.
We would like to examine the effect of adding peering links from a given policy class
in terms of the AS-level path diversity. While introducing peering links between all
pairs of IXP members with an Open policy is a reasonable approach to augment the graph,
the other cases may be less realistic. One approach is to add links between all pairs of IXP
members sharing the same peering policy. ISPs typically peer with similarly sized ISPs, 
and peering policy may be seen as a crude indicator of ISP size (with larger ISPs
being more restrictive).


\paragraph*{Results}

\tref{tab:path-diversity-by-policy} shows the results of the path diversity calculation
between 10,000 randomly selected pairs of ASes, with the selection weighted by the number of
announced addresses. We observe that especially Open links increase the mean path diversity
by 10.9\%. This is not surprising given that:
\first these are the most numerous, and \second the most likely to be missed by BGP route
collectors---relationships between the largest ISPs are most likely to be faithfully
represented in the AS relationship dataset. Restrictive and Selective links---between IXP members
with the same respective policy---have smaller impact on path diversity.

\begin{table}
\begin{minipage}{0.32\columnwidth}

\centering
\scriptsize
\tabcolsep2pt
\begin{tabular}{l*{2}r}
\toprule
 & \multicolumn{2}{c}{Path diversity} \\
\cmidrule{2-3}
Links added		& $\mu$	& $\sigma$	\\
\midrule
None			& 2.836	& 2.620		\\
Restrictive		& 2.851	& 2.633		\\
Selective		& 3.037	& 2.953		\\
Open			& 3.146 & 3.136		\\
\bottomrule
\end{tabular}
\subcaption{Policy}
\label{tab:path-diversity-by-policy}
\end{minipage}
\begin{minipage}{0.72\columnwidth}

\centering
\scriptsize
\tabcolsep2pt
\begin{tabular}{l*{4}r}
\toprule
 & \multicolumn{2}{c}{Mean path diversity} \\
\cmidrule{2-3}
Scenario				& Before	& After	& Difference\\
\midrule
Valley-free  			& 1.100 	& 1.023	& 7.03\perc\\
+ Open links 			& 1.101		& 1.023	& 7.02\perc\\
Multiple peering links 	& 1.267		& 1.267	& 0.02\perc\\
+ Open links 			& 1.499		& 1.499	& 0.04\perc\\
\bottomrule
\end{tabular}
\subcaption{Depeering}
\label{tab:path-diversity-depeering}
\end{minipage}
\caption{The effect on path diversity of \first augmenting the relationship graph with IXP
membership data (values are not cumulative), by policy and \second depeering two tier one
ISPs, by constraint type and whether or not the graph is augmented with IXP membership
data. The values are rounded off.}
\label{tab:path-diversity-policy-depeering}
\end{table}

\subsection{Effects of depeering on path diversity}
\paragraph*{Setup}
We have already noted the effect of a \emph{depeering} event, where one ISP chooses to
cease sharing its routes with one of its peers. We would like to investigate the effect of
such an event on path diversity. We synthetically create a depeering event by removing the
peering between two tier ones, in this case Deutsche Telekom (AS3320) and Level3 (AS1, AS3356
and AS3549). We then evaluate the path diversity between the \emph{exclusive} customer cones
of the ISPs, \ie the sets of customers which have only the one but not the other as 
(transit) upstream providers, and vice versa.
We examine these customer cones down to a depth of three. As before, we
also examine the effect of adding peering links. Now we also consider the possibility of a
model more liberal than valley-free, allowing to traverse multiple peering links, 
as examined also by Hu \etal~\cite{5462219} and Kotronis \etal~\cite{kotronis2429control}. 
We evaluate the effect of allowing this while still maintaining the overall
valley-free model (\ie not permitting a provider to transit via its customers).

\paragraph*{Results} 
\tref{tab:path-diversity-depeering} shows the results of the depeering for the different
scenarios involving 10,000 pairs of ASes,
selected pairwise from each of the tier one AS's exclusive customer cones. For the
valley-free model, we observe a negligible increase in path diversity by adding the extra links from PeeringDB,  
both before and after the depeering. The depeering event causes an approximately $7 \perc$ decrease in path diversity.
Conversely, if we allow multiple peering links, the addition of the extra peering links boosts path diversity by
over $18 \perc$. Here we do not observe a significant drop in path diversity due to the depeering.
Allowing multiple peering links can therefore potentially increase resilience in the face of a depeering.

\section{Related Work}
\label{sec:rel-work}

\paragraph*{Tensor Products} Soul{\'e} \etal \cite{Soule:2013:MNM:2535771.2535792} 
use a similar process to the one presented in this paper in the context of network
management. Their goal is to enforce bandwidth allocation subject to path constraints 
represented by regular expressions and consequently by NFAs. They use tensor products
in a different context than our approach, since we focus on path diversity and the implications
arising in its preservation across transformations. In particular, we further propose the 
addition of aggregator nodes, acting as inhibitors to min-cut inflation. 

\paragraph*{Network Resilience} Previous research on
resilient networks~\cite{cetinkaya:2013:MRA, cheng:2013:PGD} considers the network as 
a set of nodes and links, annotated with geographical properties.
Consequently, researchers can calculate min-cuts, path distances or shared fate 
link groups that are affected in a correlated fashion during a disaster. On the other hand, 
networks are run as policy-compliant administrative domains~\cite{caesar2005bgp};
the choice of paths that traffic can traverse is constrained. The view of the
network as a geographical map cannot capture this behavior. Thus, we argue that network resilience
should be also estimated under a policy-compliance framework. Our approach enables exactly that, allowing
to run vanilla min-cut calculation algorithms like Ford-Fulkerson on the transformed graph, with tight min-cut approximations under certain conditions. 

\paragraph*{Min-cuts with Policies} Sobrinho \etal~\cite{Sobrinho:2012:TCD:2369156.2369160}
describe a model for understanding the connectivity provided by route-vector protocols in the
face of routing policies.
Erlebach \etal~\cite{Erlebach:2004:CDP:2149938.2149945,erlebach2009connectivity}
study valid s-t-paths and s-t-cuts in the valley-free model and prove the NP-hardness of the 
vertex-disjoint min-cut problem. On the other hand, they prove that the edge-disjoint version can be solved
in polynomial time for valley-free policies; we have verified this statement in our framework. 
Both works focus on specific aspects of the general problem (route-vector protocols and valley-free policies respectively), 
while we are delving into a more general methodology for min-cut estimations. 
Sobrinho \etal~\cite{Sobrinho:2012:TCD:2369156.2369160} examine the dynamics of
a routing protocol with their work, while we concentrate on a general method to understand the 
effect of stable network policies on path diversity, ignoring for example the dynamics of routing 
convergence. Teixeira \etal~\cite{Teixeira:2003:CMP:781027.781069}
study vertex- and edge-disjoint paths in undirected Internet topology models, 
but without taking routing policies into account.

\section{Conclusion}
\label{sec:conclusion}

Path diversity and bisection bandwidth are useful
metrics to describe how resilient or rich a network is. Network
policies, imposed by network administrators and applied via routing
protocols or network configuration, can constrain the natural
path diversity of a network graph. With this work, we described and 
proved the correctness of 
a generic methodology for min-cut computations in arbitrary graphs,
assuming policies that can be formulated using regular expressions.
Our approach can be applied in a variety of scenarios, some of which
are briefly showcased in this paper. These include the investigation
of Internet topology and alternative policy models in the Internet, 
effectively studying Internet-wide resilience and the effects of
inter-AS connectivity on path diversity.
We see further potential for our approach in the analysis of MPTCP flow
path availability in data center networks, 
and path selection optimization in multipath flow 
routing applications. Achieving tighter bounds for the min-cut is another topic of interest.

\section{Acknowledgements}

We would like to thank Dr. Stefan Schmid (TU Berlin) for his valuable advice,
as well as the anonymous reviewers for their feedback.
This work has received funding from the European Research Council Grant Agreement n. 338402.

\bibliographystyle{ieeetr}
{
\bibliography{graph-transformation}

\begin{thebibliography}{10}

\bibitem{arlinghaus2001graph}
S.~L. Arlinghaus, W.~C. Arlinghaus, and F.~Harary, {\em {Graph Theory and
  Geography: an Interactive View}}.
\newblock Wiley New York, 2001.

\bibitem{Xu:2006:MMI:1159913.1159934}
W.~Xu and J.~Rexford, ``{MIRO: Multi-path Interdomain Routing},'' in {\em
  Proceedings of ACM SIGCOMM}, 2006.

\bibitem{caesar2005bgp}
M.~Caesar and J.~Rexford, ``{BGP Routing Policies in ISP Networks},'' {\em
  Network, IEEE}, vol.~19, no.~6, pp.~5--11, 2005.

\bibitem{gao2001stable}
L.~Gao and J.~Rexford, ``{Stable Internet Routing Without Global
  Coordination},'' {\em IEEE/ACM TON}, vol.~9, no.~6, pp.~681--692, 2001.

\bibitem{bruggemann1993regular}
A.~Br{\"u}ggemann-Klein, ``{Regular Expressions into Finite Automata},'' {\em
  Theoretical Computer Science}, vol.~120, no.~2, pp.~197--213, 1993.

\bibitem{ford1962flows}
L.~Ford and D.~R. Fulkerson, {\em {Flows in Networks}}, vol.~1962.
\newblock Princeton University Press, 1962.

\bibitem{Soule:2013:MNM:2535771.2535792}
R.~Soul{\'e}, S.~Basu, R.~Kleinberg, E.~G. Sirer, and N.~Foster, ``{Managing
  the Network with Merlin},'' in {\em Proceedings of ACM HotNets}, 2013.

\bibitem{caida-as-rel}
``{The CAIDA AS Relationships Dataset}.''
  \url{http://www.caida.org/data/as-relationships/}.
\newblock Dataset used: 2013-11-01.

\bibitem{4032726}
H.~Han, S.~Shakkottai, C.~Hollot, R.~Srikant, and D.~Towsley, ``{Multi-Path
  TCP: A Joint Congestion Control and Routing Scheme to Exploit Path Diversity
  in the Internet},'' {\em IEEE/ACM TON}, vol.~14, no.~6, pp.~1260--1271, 2006.

\bibitem{rfc6824}
A.~Ford, C.~Raiciu, M.~Handley, and O.~Bonaventure, ``{TCP Extensions for
  Multipath Operation with Multiple Addresses}.'' RFC 6824 (Experimental), Jan.
  2013.

\bibitem{Kang:2013:CA:2497621.2498106}
M.~S. Kang, S.~B. Lee, and V.~D. Gligor, ``{The Crossfire Attack},'' in {\em
  Proceedings of the 2013 IEEE Symposium on Security and Privacy}, 2013.

\bibitem{prince2013ddos}
M.~Prince, ``{The DDoS That Knocked Spamhaus Offline (And How We Mitigated
  It)}.'' {(Mar. 20, 2013, 06:26 PM), CloudFlare Blog}, 2013.

\bibitem{chow2007distributed}
S.~T. Chow, D.~Wiemer, and J.-M. Robert, ``{Distributed Defence Against DDoS
  Attacks},'' July~5 2007.
\newblock US Patent App. 11/822,341.

\bibitem{Peeters:2003:MEB:958763.958769}
R.~Peeters, ``{The Maximum Edge Biclique Problem is NP-complete},'' {\em
  Discrete Appl. Math.}, vol.~131, no.~3, pp.~651--654, 2003.

\bibitem{DBLP:conf/colognetwente/Kayaaslan10}
E.~Kayaaslan, ``{On Enumerating All Maximal Bicliques of Bipartite Graphs},''
  in {\em Proceedings of CTW}, 2010.

\bibitem{BinkeleRaible201064}
D.~Binkele-Raible, H.~Fernau, S.~Gaspers, and M.~Liedloff, ``{Exact
  Exponential-time Algorithms for Finding Bicliques},'' {\em Information
  Processing Letters}, vol.~111, no.~2, pp.~64 -- 67, 2010.

\bibitem{underwoodwrestling}
T.~Underwood, ``{Wrestling With the Zombie: Sprint Depeers Cogent, Internet
  Partitioned}.'' {(Oct. 31, 2008, 10:07 AM), Renesys Blog}, 2008.

\bibitem{brown2008peering}
M.~A. Brown, A.~Popescu, and E.~Zmijewski, ``{Peering Wars: Lessons Learned
  From the Cogent-Telia De-peering}.'' {Presentation, Renesys Corp, New York},
  2008.

\bibitem{Chatzis:2013:MIM:2541468.2541473}
N.~Chatzis, G.~Smaragdakis, A.~Feldmann, and W.~Willinger, ``{There is More to
  IXPs Than Meets the Eye},'' {\em SIGCOMM CCR}, vol.~43, no.~5, pp.~19--28,
  2013.

\bibitem{euroix-ixp-list}
{EuroIX: European Internet Exchange Association}, ``{List of Known IXPS Around
  the Globe}.'' \url{https://www.euro-ix.net/resources-list-of-ixps}, 2014.

\bibitem{caida-routeviews}
``{Routeviews Prefix to AS mappings Dataset for IPv4 and IPv6}.''
  \url{http://www.caida.org/data/routing/routeviews-prefix2as.xml}.
\newblock Dataset used: 2014-01-11, 12:00 UTC.

\bibitem{Ager:2012:ALE:2377677.2377714}
B.~Ager, N.~Chatzis, A.~Feldmann, N.~Sarrar, S.~Uhlig, and W.~Willinger,
  ``{Anatomy of a Large European IXP},'' {\em SIGCOMM CCR}, vol.~42, no.~4,
  pp.~163--174, 2012.

\bibitem{peeringdb-routing-ecosystem}
A.~Lodhi, N.~Larson, A.~Dhamdhere, C.~Dovrolis, {\em et~al.}, ``{Using
  peeringDB to Understand the Peering Ecosystem},'' {\em ACM SIGCOMM CCR},
  vol.~44, no.~2, pp.~20--27, 2014.

\bibitem{5462219}
C.~Hu, K.~Chen, Y.~Chen, and B.~Liu, ``{Evaluating Potential Routing Diversity
  for Internet Failure Recovery},'' in {\em Proceedings of IEEE INFOCOM}, 2010.

\bibitem{kotronis2429control}
V.~Kotronis, X.~Dimitropoulos, R.~Kl{\"o}ti, B.~Ager, P.~Georgopoulos, and
  S.~Schmid, ``{Control Exchange Points: Providing QoS-enabled End-to-End
  Services via SDN-based Inter-domain Routing Orchestration},'' in {\em
  Proceedings of the Research Track of ONS 2014}, 2014.

\bibitem{cetinkaya:2013:MRA}
E.~K. \c{C}etinkaya, M.~J. Alenazi, A.~M. Peck, J.~P. Rohrer, and J.~P.~G.
  Sterbenz, ``{Multilevel Resilience Analysis of Transportation and
  Communication Networks},'' {\em Springer Telecommunication Systems Journal},
  2013.

\bibitem{cheng:2013:PGD}
Y.~Cheng, J.~Li, and J.~P.~G. Sterbenz, ``{Path Geo-diversification: Design and
  Analysis},'' in {\em Proceedings of IEEE/IFIP RNDM Workshop}, 2013.

\bibitem{Sobrinho:2012:TCD:2369156.2369160}
J.~a.~L. Sobrinho and T.~Quelhas, ``{A Theory for the Connectivity Discovered
  by Routing Protocols},'' {\em IEEE/ACM Trans. Netw.}, vol.~20, no.~3,
  pp.~677--689, 2012.

\bibitem{Erlebach:2004:CDP:2149938.2149945}
T.~Erlebach, A.~Hall, A.~Panconesi, and D.~Vukadinovi\'{c}, ``{Cuts and
  Disjoint Paths in the Valley-free Path Model of Internet BGP Routing},'' in
  {\em Proceedings of CAAN}, 2005.

\bibitem{erlebach2009connectivity}
T.~Erlebach, L.~S. Moonen, F.~C. Spieksma, and D.~Vukadinovic, ``{Connectivity
  Measures for Internet Topologies on the Level of Autonomous Systems},'' {\em
  Operations research}, vol.~57, no.~4, pp.~1006--1025, 2009.

\bibitem{Teixeira:2003:CMP:781027.781069}
R.~Teixeira, K.~Marzullo, S.~Savage, and G.~M. Voelker, ``{Characterizing and
  Measuring Path Diversity of Internet Topologies},'' in {\em Proceedings of
  ACM SIGMETRICS}, 2003.

\end{thebibliography}
}

\end{document}